\documentclass[UTF8,12pt]{article}
\usepackage{amsmath}
\usepackage{amssymb}
\usepackage{listings}
\usepackage{xcolor}
\usepackage{graphicx}
\usepackage{booktabs} 
\usepackage{caption}
\usepackage{geometry}
\usepackage{array}
\usepackage{amsmath}
\usepackage{url}
\usepackage[caption=false]{subfig}
\usepackage{longtable}
\usepackage{abstract}
\usepackage{amsthm}
\usepackage{authblk}
\usepackage[numbers,sort&compress]{natbib}
\theoremstyle{plain} \newtheorem{theorem}{Theorem}[section]
\theoremstyle{definition} 
\theoremstyle{plain} \newtheorem{corollary}[theorem]{Corollary}
\theoremstyle{remark}\newtheorem*{remark}{Remark}
\theoremstyle{plain} \newtheorem{lemma}[theorem]{Lemma}

\title{Analysis of a stochastic SIR model with media effects}

\author[a]{Jiaxun Li}
\author[a]{Yanni Xiao}
\affil[a]{School of Mathematics and Statistics, Xi'an Jiaotong University, Xi'an 710049, P.R. China}

\date{}
\begin{document}

\captionsetup[figure]{font={small,stretch=1.1},labelfont={bf},name={Fig.},labelsep=period}
\maketitle
\begin{abstract}
	In this study, we investigate a stochastic SIR model with media effects. The uniqueness and the existence of a global positive solution are studied. The sufficient conditions of extinction and persistence of the disease are established.  We obtain the  basic reproduction number $R_0^S$ for stochastic system, which can act as the threshold given small environmental noise. Note that large noise can induce the disease extinction with probability of 1,  suggesting that environmental noises can not be ignored when investigating threshold dynamics. Further, inclusion of media induced behaviour changes does not affect the threshold itself, which is similar to the conclusion of the deterministic models.  However, numerical simulations suggest that media impacts induce the disease infection decline.
\end{abstract}
\section{Introduction}

Since the pioneer work of Kermack and McKendrick\cite{model}, mathematical models have played an important role in investigating epidemics in the real world. In the classical endemic models, the incidence rate is assumed to be bilinear with the form $\beta SI$, where $\beta$ is a positive constant represents the probability of transmission per contact. However, when a disease appears and breaks out, people always take protective measures spontaneously influenced by surroundings or the mass media, which may mitigate the spread of the disease. Examples of such media influence include the spread of the 2003 SARS, the 2009 H1N1, and the recent COVID-19\cite{Media-endemic1,Media-endemic2,Media-endemic3,Media-endemic4,Media-Covid1,Media-Covid2}. Hence it is unreasonable to assume $\beta$ as a constant.

As a result, many models were proposed in which the impact of media coverage on disease spread is considered. Liu et al., in \cite{media-model1}, described the media effect by multiplying the transmission coefficient $\beta$ with $exp(-a_1E-a_2I-a_3H)$, where $E, I$ and $H$ are the numbers of reported exposed, infectious and hospitalized individuals, respectively. Li et al., in \cite{Media-endemic2}, proposed an SIS model with incidence rate $(\beta_1-\beta_2\frac{I}{m+I})\frac{SI}{N}$ to reflect the reduction of contact rate through media coverage. Cui et al.\cite{media-model2.1}, Wang and Xiao\cite{media-model2.2}, Song and Xiao\cite{media-model2.3} used the incidence rate $\beta exp({-\alpha I})SI$ to approximate the impact of media coverage and have proposed various models with different assumptions.

A common assumption for the models in \cite{Media-endemic2,media-model1,media-model2.1,media-model2.2,media-model2.3} is that the spread of disease is a definite process, while in the real world, epidemics will fluctuate inevitably due to the environmental white noise. Hence adding stochastic factors to epidemic models will be a meaningful approach. In fact, many stochastic models for epidemics have been developed. For example, Tornatore et al., in \cite{SDE1}, discussed an stochastic SIR system with and without delays. A Gray et al., in \cite{SDE2}, developed an stochastic SIS system and established the conditions for extinction and persistence of $I(t)$. Zhao, in \cite{SDE3}, introduced a stochastic SIR model with saturated incidence and gave the threshold of the system. There are many other stochastic models under various assumption, see \cite{SDE4,SDE5,SDE6,SDE7,SDE8,SDE9}. Little is known about transmission dynamics of the epidemic model with media impact and environmental perturbations, and consequently it is essential to examine how environmental stochastic factor and media impact influence the transmission dynamics of infectious diseases.

We introduce random perturbations to the following SIR model with media effects\cite{media-model2.3}.
\begin{equation}
	\left\{
		\begin{array}{l}
			dS=(\Lambda-\beta e^{-\alpha I}SI-\mu S)dt,
			\\dI=[\beta e^{-\alpha I}SI-(\mu+\gamma)I]dt,
			\\dR=(\gamma I-\mu R)dt,
	\end{array}\right.
    \label{deterministic}
\end{equation}

where $S(t), I(t), R(t)$ represents the number of susceptible, infected and recovered individuals respectively. $\Lambda$ stands for the rate of flow into the population, $\mu$ is the natural death rate, $\beta$ denotes the transmission rate, $\gamma$ represents the recovery rate and $e^{-\alpha I}$ is the reduction of transmission rate caused by media effects. All the parameters here are positive. Note that the dynamic behavior of (\ref{deterministic}) has been analyzed in detail by Song and Xiao\cite{media-model2.3}. They found the basic reproduction number $R_0^D$ defined by
\[R_0^D=\frac{\Lambda\beta}{\mu(\mu+\gamma)},\]
is a threshold of the model (\ref{deterministic}), namely, the disease-free equilibrium is globally asymptotically stable if $R_0^D\le 1$, while the endemic equilibrium is feasible and globally asymptotically stable if $R_0^D>1$.

We assume that noises in the environment will mainly affect the transmission coefficient $\beta$, as in \cite{SDE1,SDE2,SDE3}, so
\[\beta dt\to \beta dt+\sigma dB(t),\]
where $B(t)$ is a brownian motion and $\sigma$ is a positive constant, thus the deterministic model (\ref{deterministic}) is transformed to the following stochastic model:
\begin{equation}
    \left\{\begin{array}{l}dS=(\Lambda-\beta e^{-\alpha I}SI-\mu S)dt-\sigma e^{-\alpha I}SIdB(t)
		\\dI=[\beta e^{-\alpha I}SI-(\mu+\gamma)I]dt+\sigma e^{-\alpha I}SIdB(t)\\dR=(\gamma I-\mu R)dt\end{array}\right.
    \label{main}
\end{equation}
In this paper, we investigate the dynamics of system (\ref{main}), and give the conditions to determine the extinction and persistence of the disease.

The structure of this paper is organized as follows: In Section 2, we study dynamical behaviors of system (\ref{main}). In Section 3, we give some numerical examples to show the complicated stochastic dynamics of the model. We then conclude our work in Section 4.

\section{Analysis for the stochastic model}
In this paper, we let $(\Omega,\mathcal F,\{\mathcal F_t\}_{t\ge 0},\mathbb{P})$ be a complete probability space with a filtration $\{\mathcal F_t\}_{t\ge 0}$ satisfying the usual conditions, namely, it is increasing and right continuous with $\mathcal F_0$ contains all $\mathbb{P}$-null sets. Let $B(t)$ be a 1-dimensional Brownian motion defined on $(\Omega,\mathcal F,\{\mathcal F_t\}_{t\ge 0},\mathbb{P})$. We use $a\wedge b$ to denote $\min(a, b)$ and $a\vee b$ to denote $\max(a, b)$.

\subsection{Existence and uniqueness of global positive solution}
For the deterministic model (\ref{deterministic}), we know that a solution $(S(t),I(t),R(t))\in \mathbb R^3_+,\forall t\ge 0$ whenever $(S(0), I(0),R(0))\in \mathbb R^3_+$. In order for the stochastic differential equation(SDE) model (\ref{main}) to make sense, we need to show a solution of it satisfies this property as well.

\begin{theorem}
	For any given initial value $(S(0),I(0),R(0))\in\mathbb R^3_+$, the SDE (\ref{main}) has a unique global solution $(S(t), I(t),R(t))\in \mathbb R^3_+$ for all $t\ge 0$ with probability one, namely,
	\[\mathbb P\{(S(t),I(t),R(t))\in \mathbb R^3_+,\forall t\ge 0\}=1.\]
	\label{thm-existence}
\end{theorem}
\begin{proof}
	It is easy to show that the SDE (\ref{main}) satisfies the local Lipschitz condition, so for any given initial value $(S(0),I(0),R(0))\in\mathbb R^3_+$, there is a unique maximal local solution $(S(t),I(t),R(t))$ on $t\in [0,\tau_e)$, where $\tau_e$ is the explosion time. (see Theorem 2.8 in p155 in \cite{SDE}). Set

	\[\varGamma_k=\{(x,y,z)\in \mathbb R^3_+|1/k<x,y,z<k\}.\]

	Let $k_0>0$ be sufficient large so that $(S(0),I(0),R(0))\in\varGamma_{k_0}$. For each $k\ge k_0$, let
	\[\tau_k=\inf\{t\in[0,\tau_e)|(S(t),I(t),R(t))\notin \varGamma_k\}.\]
	where we set $\inf\emptyset=\infty$. $\tau_k$ is increasing as $k\to\infty$. Let $\tau_\infty=\lim_{k\to\infty}\tau_k$, whence $\tau_\infty\le\tau_e$ a.s. If we can show that $\tau_\infty=\infty$ a.s., then $\tau_e=\infty$ a.s. and $(S(t),I(t),R(t))\in \mathbb R^3_+,\forall t\ge 0$ a.s. So to complete the proof all we need to show is that $\tau_\infty=\infty$ a.s. If the statement is false, then there is a pair of constants $T>0$ and $\epsilon\in(0,1)$ such that
	\[\mathbb P\{\tau_\infty\le T\}>\epsilon.\]
	Then there exists an integer $k_1>k_0$, such that
	\begin{equation}
		\mathbb P\{\tau_k\le T\}>\epsilon, \forall k\ge k_1.
		\label{omega_k}
	\end{equation}
	Let $N(t)=S(t)+I(t)+R(t)$, from (\ref{main}), we know that
	\begin{equation}
		dN(t)=(\Lambda-\mu N)dt.
		\label{thm3.1-N}
	\end{equation}
	Solve the equation (\ref{thm3.1-N}), we get
	$N(t)=\frac\Lambda \mu +ce^{-\mu t}$,
	where $c\in \mathbb R$. Thus
	\[N(t)\le N(0)\vee \frac \Lambda\mu.\]
	For all $k\ge0$ and $t\in [0,\tau_k)$, since $S(t),I(t),R(t)>0$, we must have
	\begin{equation}
		S(t)\leq N(0)\vee \dfrac{\Lambda}{\mu}, I(t)\leq N(0)\vee \dfrac{\Lambda}{\mu}, R(t)\leq N(0)\vee \dfrac{\Lambda}{\mu}, \quad\forall k\ge 0, t\in [0,\tau_k).
		\label{thm3.1-S,I,R}
	\end{equation}
	Define a function $V:\mathbb R^3_+\to \mathbb R_+$ by
	\[V(S,I,R)=\frac 1 S+\frac 1 I+\frac 1 R.\]
	Make use of the It\^{o}'s formula(see \cite{SDE}), we have, for any $t\in [0,T]$ and $k\ge k_1$,
	\begin{equation}
		\mathbb EV(S(t\wedge\tau_k),I(t\wedge\tau_k),R(t\wedge \tau_k))=V(S(0),I(0),R(0))+\mathbb E\int_0^{t\wedge \tau_k}LV(S(s),I(s),R(s))ds,
		\label{EV}
	\end{equation}
	where
		\begin{align*}
			LV(S,I,R)={}&-\frac 1 {S^2}(\Lambda-\beta e^{-\alpha I}SI-\mu S)-\frac 1 {I^2}[\beta e^{\alpha I}SI-(\mu+\gamma)I]\\
			&-\frac 1 {R^2}(\gamma I-\mu R)+\sigma^2e^{-2\alpha I}I^2S^2\left(\frac 1{I^3}+\frac 1 {S^3}\right)\\
			={}&\frac{\beta e^{-\alpha I}I}{S}+\frac \mu S+\frac {\mu+\gamma} I+\frac{\mu} R+\frac {\sigma^2 e^{-2\alpha I}I^2} S+\frac{\sigma^2 e^{-2\alpha I}S^2} I\\
			&-\left(\frac{\Lambda}{S^2}+\frac{\beta e^{-\alpha I}S} I+\frac {\gamma I}{R^2}\right)\\
			\le{} &\frac{\beta e^{-\alpha I}I}{S}+\frac \mu S+\frac {\mu+\gamma} I+\frac{\mu} R+\frac {\sigma^2 e^{-2\alpha I}I^2} S+\frac{\sigma^2 e^{-2\alpha I}S^2} I\\
			\le{}&\frac {\beta (\frac\Lambda\mu\vee N(0))}{S}+\frac \mu S+\frac {\mu+\gamma} I+\frac \mu R+\frac {\sigma^2(\frac\Lambda\mu\vee N(0))^2}S+\frac {\sigma^2(\frac\Lambda\mu\vee N(0))^2} I\\
			\le{}&CV(S,I,R),
		\end{align*}
	where we used (\ref{thm3.1-S,I,R}) in the penultimate inequality and $C=\beta\left(\frac{\Lambda}{\mu}\vee N(0)\right)+2\mu+\gamma+\sigma^2(\frac\Lambda\mu\vee N(0))^2$. Substituting this into (\ref{EV}),
	\[\mathbb EV(S(t\wedge\tau_k),I(t\wedge\tau_k),R(t\wedge \tau_k))\le V(S(0),I(0),R(0))+C\int_0^{t\wedge \tau_k}\mathbb EV(S(s),I(s),R(s))ds.\]
	By the Gronwall inequality,
	\begin{equation}
		\mathbb EV(S(T\wedge\tau_k),I(T\wedge\tau_k),R(T\wedge\tau_k))\le V(S(0),I(0),R(0))e^{CT}.
		\label{gronwall}
	\end{equation}
	Set $\Omega_k=\{\tau_k\le T\}$ for $k\ge k_1$, by (\ref{omega_k}), we have $\mathbb P(\Omega_k)\ge \epsilon$. Note that when $k$ is large enough, by (\ref{thm3.1-S,I,R}) , for every $\omega\in\Omega_k$, at least one of $S(\tau_k, \omega), I(\tau_k, \omega) \text{ and } R(\tau_k, \omega)$ equals $1/k$, hence
	\[V(S(\tau_k,\omega),I(\tau_k,\omega),R(\tau_k,\omega))\ge k.\]
	It then follows from(\ref{gronwall}) that
	\[V(S(0),I(0),R(0))e^{CT}\ge \mathbb E\left[I_{\Omega_k}(\omega)V(S(\tau_k,\omega),I(\tau_k,\omega),R(\tau_k,\omega))\right]\ge k \mathbb P(\Omega_k)\ge \epsilon k.\]
	Letting $k\to\infty$ leads to the contradiction
	\[\infty>V(S(0),I(0),R(0))e^{CT}=\infty,\]
	so we must have $\tau_\infty=\infty$ a.s., thus the proof is complete.
\end{proof}
Let
\[\varGamma=\{(x,y,z)\in\mathbb R_+^3 | x+y+z<\frac\Lambda \mu\}.\]
In the rest of this paper, we assume that $(S(0),I(0),R(0))\in\varGamma$. By Theorem 2.1 and (\ref{thm3.1-S,I,R}), we have the following corollary:
\begin{corollary}
	For any given initial value $(S(0),I(0),R(0))\in\varGamma$, the SDE(\ref{main}) has a unique global solution $(S(t),I(t),R(t))\in\varGamma$ for all $t\ge0$ a.s.
\end{corollary}

\subsection{Extinction}

In this section, we deduce the condition under which the disease dies out. Define
\[R_{0}^S=\frac{\beta \Lambda}{\mu(\mu+\gamma)}-\frac{\sigma^2\Lambda^2}{2\mu^2(\mu+\gamma)}\]
be the basic reproduction number for SDE model (\ref{main}). The next theorem shows that this parameter has the similar property as $R_0^D$ for the deterministic model (\ref{deterministic}).
\begin{theorem}
	If
	\begin{equation}
		\label{extinct-condition1}
		R_{0}^S<1 \text{ and } \sigma^2<\frac{\mu\beta}\Lambda,
	\end{equation}
	then for any given initial value $(S(0),I(0),R(0))\in \varGamma$, the solution of SDE(\ref{main}) obeys
	\begin{equation}
		\label{Extinction1}
		\limsup_{t\to\infty}\frac 1 t\ln(I(t))\le (\mu+\gamma)(R_0^S-1)<0\qquad\text{a.s.,}
	\end{equation}
	namely, the disease $I(t)$ will die out exponentially with probability one. Moreover, we have
	\label{thm-extinct1}
	\begin{gather}
		\label{cor4.2.1}
		\limsup_{t\to\infty}\frac 1 t\ln(R(t))\le\max\{(\mu+\gamma)(R_0^S-1),-\mu\}<0\qquad\text{a.s.,}\\
		\label{cor4.2.2}
		\limsup_{t\to\infty}\frac 1 t\ln(\frac\Lambda\mu-S(t))\le\max\{(\mu+\gamma)(R_0^S-1),-\mu\}<0\qquad\text{a.s.,}
	\end{gather}
	which means $R(t)$ and $\frac\Lambda\mu-S(t)$ will tend to zero exponentially.
\end{theorem}
\begin{proof}
	By the It\^o's formula, we have
	\begin{equation}
		\label{ito-extinction1}
		\ln(I(t))=\ln(I(0))+\int^t_0f(S(s),I(s))ds+\int^t_0\sigma e^{-\alpha I}S(s)dB(s),
	\end{equation}
	where
	\begin{equation}
		\label{thm4.1-f}
		f(S,I)=\beta e^{-\alpha I}S-(\mu+\gamma)-\frac 1 2\sigma^2e^{-2\alpha I}S^2.
	\end{equation}
	Consider the quadratic function
	\begin{equation}
		\label{theorem4.1-g}
		g(x)=-\frac 1 2 \sigma^2x^2+\beta x-\mu-\gamma.
	\end{equation}

	Note that $g(x)$ attaches its maximum value at $x_0=\frac \beta {\sigma^2}$. From (\ref{extinct-condition1}), we have $x_0>\frac\Lambda\mu$. So $g(x)$  increases in $(0,\frac\Lambda\mu)$. Thus by Corollary 2.2 and (\ref{extinct-condition1}), we have
	\[f(S,I)=g(e^{-\alpha I}S)\le g(\frac\Lambda\mu)=\frac{\beta\Lambda}\mu -\mu-\gamma-\frac 1 2 \frac{\sigma^2\Lambda^2}{\mu^2}=(\mu+\gamma)(R_0^S-1)\qquad a.s. \]
	It follows from (\ref{ito-extinction1}) that
	\[\ln(I(t))\le \ln(I(0))+(\mu+\gamma)(R_0^S-1)t+\int^t_0 \sigma e^{-\alpha I(s)}S(s)dB(s)\qquad a.s.,\]
	which implies
	\[\limsup_{t\to\infty}\frac 1 t \ln(I(t))\le(\mu+\gamma)(R_0^S-1)+\limsup_{t\to\infty}\frac 1 t \int^t_0\sigma e^{-\alpha I(s)}S(s)dB(s)\qquad a.s.\]
	However, by the large number theorem for martingales (see \cite{SDE}), we have
	\[\limsup_{t\to\infty}\frac 1 t\int^t_0\sigma e^{-\alpha I(s)}(S(s))dB(s)=0\qquad a.s.\]
	We therefore complete the proof of (\ref{Extinction1}). Now set
	\[m:=-(\mu+\gamma)(R_0^S-1).\]
	From (\ref{Extinction1}) we know that
	\[\limsup_{t\to\infty}\frac 1 t\ln(I(t))\le-m<0\qquad a.s.,\]
	which means there exists a set $\Omega_I$ such that $\mathbb P(\Omega_I)=1$ and for every sufficiently small $\varepsilon>0$ and $\omega\in \Omega_I$, there exists a $T=T(\varepsilon,\omega)>0$, such that
	\begin{equation}
		I(t,\omega)<e^{(-m+\varepsilon)t}\qquad\forall t>T.
	\label{cor4.2-I}
	\end{equation}
	Substituting this into SDE(\ref{main}) we have
	\begin{equation}
		\label{cor4.2-compare}
		dR(t,\omega)<(\gamma e^{(-m+\varepsilon)t}-\mu R(t,\omega))dt\qquad\forall t>T.
	\end{equation}
	Consider the deterministic differential equation
	\begin{equation*}
		dx=(\gamma e^{(-m+\varepsilon)t}-\mu x)dt,
	\end{equation*}
	a simple calculation shows its general solution is
	\[x=\frac{\gamma}{-m+\varepsilon+\mu}e^{(-m+\epsilon)t}+ae^{-\mu t},\qquad a\in \mathbb R.\]
	By the comparing principle, the solution of (\ref{cor4.2-compare}) satisfies
	\begin{equation}
	   R(t,\omega)<\frac{\gamma}{-m+\varepsilon+\mu}e^{(-m+\epsilon)t}+ae^{-\mu t},\qquad a\in \mathbb R,t>T,
	   \label{cor4.2-R}
	\end{equation}
	which means
	\[\limsup_{t\to\infty}\frac 1 t \ln(R(t,\omega))\le\max\{-m+\epsilon,-\mu\}.\]
	Letting $\epsilon\to 0$ lead to the assertion (\ref{cor4.2.1}). Then (\ref{cor4.2.2}) is an immediate result by combining (\ref{cor4.2-I}),(\ref{cor4.2-R}) and the fact $\lim_{t\to\infty}R(t)+I(t)+S(t)=\frac\Lambda\mu$.
\end{proof}

Theorem 2.3 shows that the disease will die out if $R_{0}^S<1$ and the intensity of stochastic perturbation is relatively small. The next theorem covers the case when the intensity of stochastic perturbation is rather large.
\begin{theorem}
	If
	\begin{equation}
		\sigma^2>\frac{\beta^2}{2(\mu+\gamma)},
		\label{thm4.2-condition}
	\end{equation}
	then for any given initial value $(S(0),I(0),R(0))\in \varGamma$, the solution of the SDE model (\ref{main}) obeys
	\begin{equation}
		\label{thm4.3-assertion}
		\limsup_{t\to\infty}\frac 1 t\ln(I(t))\le \frac{\beta^2}{2\sigma^2}-\mu-\gamma<0\qquad\text{a.s.},
	\end{equation}
	namely, $I(t)$ will die out exponentially with probability one. Moreover,
	\begin{gather}
		\label{cor4.4.1}
		\limsup_{t\to\infty}\frac 1 t\ln(R(t))\le\max\{\frac{\beta^2}{2\sigma^2}-\mu-\gamma,-\mu\}<0\qquad\text{a.s.,}\\
		\label{cor4.4.2}
		\limsup_{t\to\infty}\frac 1 t\ln(\frac\Lambda\mu-S(t))\le\max\{\frac{\beta^2}{2\sigma^2}-\mu-\gamma,-\mu\}<0\qquad\text{a.s.,}
	\end{gather}
	which means $R(t)$ and $\frac\Lambda\mu-S(t)$ will tend to zero exponentially.
	\label{thm-extinct2}
\end{theorem}

	\begin{proof}
		Since the proofs of (\ref{cor4.4.1}) and (\ref{cor4.4.2}) are very similar to that of (\ref{cor4.2.1}) and (\ref{cor4.2.2}), we only prove the inequality (\ref{thm4.3-assertion}) here. We use the same notation as in the proof of Theorem 2.3. Note that
		\[g(x_0)=-\mu-\gamma+\frac {\beta^2}{2\sigma^2},\]
		so by Theorem 2.1 and (\ref{thm4.2-condition}), we have
		\[f(S,I)=g(e^{-\alpha I}S)<g(x_0)=-\mu-\gamma+\frac {\beta^2}{2\sigma^2},\]
		which is negative by condition (\ref{thm4.2-condition}). This implies, in the same way as in the proof of Theorem 2.3 that
		\[\limsup_{t\to\infty}\frac 1 t\ln(I(t))\le -\mu-\gamma+\frac{\beta^2}{2\sigma^2}\qquad a.s.\]
		as required.
	\end{proof}
\begin{remark}
	Note that the condition (\ref{thm4.2-condition}) implies $R_0^S<1$ since
	\[R_0^S<\frac{\beta\Lambda}{\mu(\mu+\gamma)}-\frac{\beta^2\Lambda^2}{4\mu^2(\mu+\gamma)^2}\le 1.\]
	Thus when the stochastic perturbation is large enough, $R_0^S$ will automatically go below 1.
\end{remark}
\begin{remark}
	Theorem 2.3 and 2.4 are also true if we only suppose $(S(0),I(0),R(0))\in\mathbb R_+^3.$ Since if $N(0)>\frac\Lambda\mu$, we have $N(t)\to\frac\Lambda\mu,t\to\infty$ and $N(t)$ decreases, then for every $\epsilon>0$, $N(t)<\frac\Lambda\mu+\epsilon$ for large enough $t$. Therefore we can repeat the proof of Theorem 2.3 and 2.4 by first replacing $\frac\Lambda\mu$ with $\frac\Lambda\mu+\epsilon$ and then letting $\epsilon\to 0$.
\end{remark}

\subsection{Persistence}

In this section, we deduce the weak persistence and the mean persistence of $I(t)$. We first pay attention on the variable $Q(t):=\frac \Lambda\mu-S(t)$ and deduce its persistence. Note that from our model equations we have $N(t)\to\frac\Lambda\mu,t\to\infty$, then
\[\lim_{t\to\infty}\frac{Q(t)}{I(t)+R(t)}=1.\]
Use this fact, combine with the persistence of $Q(t)$, we can derive the weak persistence of $I(t)$. Recall the definition of $f$ in (\ref{thm4.1-f}) and $g$ in (\ref{theorem4.1-g}). Now we set
\begin{equation}
	\tilde{f}(Q,I)=f(N-Q,I)=\beta e^{-\alpha I}(\frac\Lambda\mu-Q)-(\mu+\gamma)-\frac 1 2 \sigma^2 e^{-2\alpha I}(\frac\Lambda\mu-Q)^2.
	\label{Persistence-f}
\end{equation}
Define $h$ as
\begin{equation}
	h(x,y):=(\frac\Lambda\mu-x)e^{-\alpha y},
	\label{Persistence-h}
\end{equation}
thus
\[\tilde f(Q,I)=g(h(Q,I)).\]
By Corollary 2.2, $(Q(t),I(t),R(t))\in \varGamma'$ for all $t\ge 0$ a.s. if $(Q(0),I(0),R(0))\in \varGamma'$, where
\[\varGamma'=\{(x,y,z)\in \mathbb R^3_+ | y+z<x<\frac\Lambda\mu\}.\]
Hence we may assume the domain of $\tilde f$ and $h$ be
\[\varGamma'_0=\{(x,y)\in\mathbb R^2_+ | y<x<\frac\Lambda\mu\}.\]
\begin{theorem}
	If $R_{0}^S>1$, then for any given initial value $(S(0),I(0),R(0))\in \varGamma$, the solution of the SDE (\ref{main}) obeys
	\begin{equation}
		\label{thm5.1-conclusion-1}
		\limsup_{t\to\infty} Q(t)\ge\xi\qquad\text{a.s.},
	\end{equation}
	where $\xi$ is the unique root in $(0,\frac\Lambda\mu)$ of the equation
	\begin{equation}
		\beta e^{-\alpha \xi}(\frac\Lambda\mu-\xi)-(\mu+\gamma)-\frac 1 2\sigma^2e^{-2\alpha \xi}(\frac\Lambda\mu-\xi)^2=0.
		\label{thm5.1-equation-1}
	\end{equation}
	Futhermore, we have
	\begin{equation}
		\label{thm5.3-conclusion1}
		\limsup_{t\to\infty} I(t)\ge(1+\frac\gamma\mu)^{-1}\xi\qquad\text{a.s.}
	\end{equation}
	Namely, the disease $I(t)$ will have weak persistence with probability one.	
	\label{thm-persist}

\end{theorem}
To prove this, we need some properties of the function $\tilde f$ in (\ref{Persistence-f}) and $h$ in (\ref{Persistence-h}), which are summarized in the below lemma. Fig.\ref{lemma3.8} illustrates the first two assertions of Lemma 2.6.

\begin{figure}[!htbp]
    \centering
    \includegraphics[width=0.5\linewidth,trim= 200 130 250 100,clip]{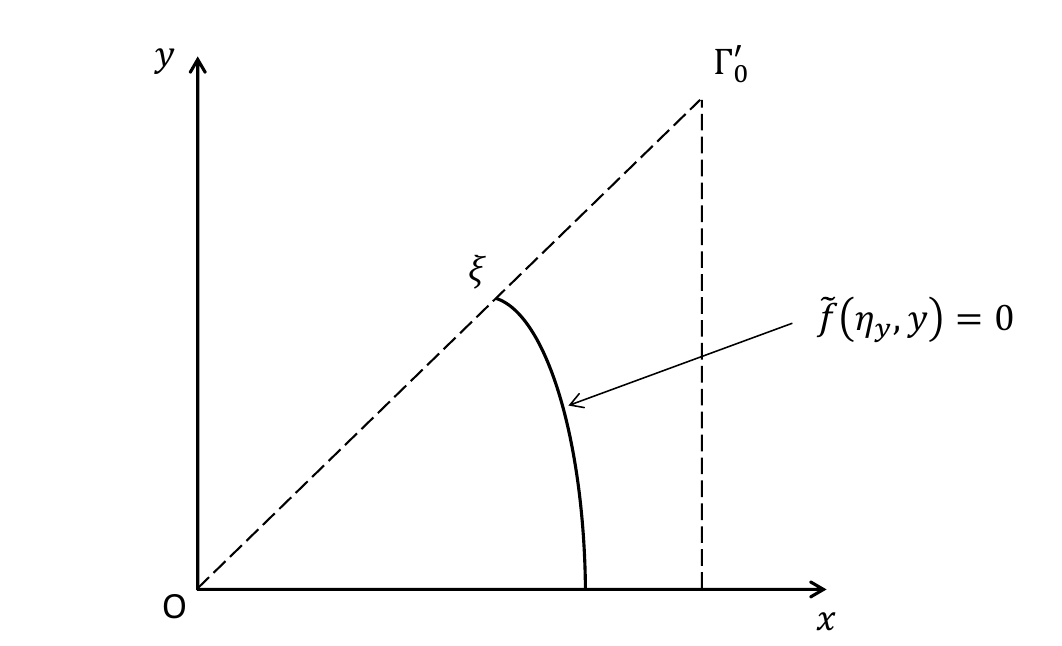}
    \caption{Illustration of 1) and 2) in Lemma2.6}
	\label{lemma3.8}
\end{figure}

\begin{lemma}
	If $R_0^S>1$, functions $\tilde f$ in (\ref{Persistence-f}) and $h$ in (\ref{Persistence-h}) have the following properties:
	\begin{enumerate}
		\item There is a unique root $\xi$ in $(0,\frac\Lambda\mu)$ of the equation $\tilde f(\xi,\xi)=0$.
		\item For each $y\in [0,\xi]$, the equation $\tilde f(\eta_y,y)=0$ has a unique root $\eta_y$ in $[y,\frac\Lambda\mu)$.
		\item For each $(x,y)\in \varGamma_0'$, $\tilde f(x,y)>0$ if $y\in(0,\xi), x\in(y,\eta_y)$. $\tilde f(x,y)<0$ if $y\in(0,\xi), x\in(\eta_y,\frac\Lambda\mu)$ or $y\in[\xi,\frac\Lambda\mu)$.
		\item For each $(x,y)\in \varGamma_0'$, $h(x,x)<h(x,y)<h(x,0)<\frac\Lambda\mu$.
	\end{enumerate}
\end{lemma}
\begin{proof}
	Consider the equation $g(x)=0$, where $g$ is defined in (\ref{theorem4.1-g}). If $R_0^S>1$, we have
	\begin{equation*}
		\beta>\dfrac{\sigma^2\Lambda}{2\mu}+\dfrac{\mu(\mu+\gamma)}{\Lambda}\ge \sqrt{2\sigma^2(\mu+\gamma)},
	\end{equation*}
	which induces $g(x)=0$ has two roots $x_1,x_2$. Denote the smaller one as $x_1$, then
	\begin{equation}
		0<x_1<\frac\Lambda\mu<x_2,
		\label{lemma5.2-inequality}
	\end{equation}
	which means $g(x)=0$ has a unique root $x_1$ in $(0,\frac\Lambda\mu)$. Consider the function
	\begin{equation*}
		r(x):=h(x,x)=e^{-\alpha x}(\frac\Lambda\mu-x).
	\end{equation*}
	Note that $r(x)$ is decreasing in $[0,\frac\Lambda\mu]$, $r(0)=\frac\Lambda\mu$ and $r(\frac\Lambda\mu)=0$, thus $r(x)=x_1$ has a unique root $x=\xi$ and $r(x)=x_2$ has no root in $[0,\frac\Lambda\mu]$, which indicates the following equation
	\begin{equation*}
		\tilde f(\xi,\xi)=g(r(\xi))=0
	\end{equation*}
	has a unique root $\xi$ in $(0,\frac\Lambda\mu)$. Therefore the proof of 1) is complete. To prove 2), note that
	\begin{equation*}
		\tilde f(x,y)=0\Leftrightarrow h(x,y)=x_1\text{ or }x_2.
	\end{equation*}
	On one hand, for each $(x,y)\in \overline {\varGamma'_0}$, $ h(x,y)\in [0,\frac\Lambda\mu]$, thus $h(x,y)< x_2$. On the other hand, it is easy to see that for each $y\in [0,\frac\Lambda\mu]$, the equation $h(x,y)=x_1$ has a unique root $x=\eta_y$ in $[y,\frac\Lambda\mu)$ if and only if $y\in [0,\xi]$, thus the second assertion is true.

	To prove 3), notice that for each fixed $y$, $\tilde f(x,y)$ is a quadratic function of $x$. Then we only need to prove that $\tilde f(y,y)>0$ for each $y<\xi$ and $\tilde f(y,y)<0$ for each $y>\xi$. To show this, recall $\tilde f(y,y)=g(r(y))$, $r(0)=\frac\Lambda\mu$, $r(\xi)=x_1$, $r(\frac\Lambda\mu)=0$. Then this proposition is an immediate result by (\ref{lemma5.2-inequality}) and the fact that $r$ is a decreasing function. Assertion 4) is also an immediate result for the fact that for each $x$, $h(x,y)$ decreases as $y$ increases.
\end{proof}
\begin{proof}[proof of Theorem2.5]
	Note that in Lemma 2.6 we have already proved the existence and uniqueness of $\xi$. We now begin to prove the weak persistence of the variable $Q(t)$. If it is not true, then there is a sufficiently small $\epsilon>0$ such that $\mathbb P(\Omega_1)>\epsilon$, where $\Omega_1=\{\omega|\limsup_{t\to\infty}Q(t,\omega)\le\xi-2\epsilon\}$. Let $\epsilon$ be sufficiently small such that
	\begin{equation}
		\quad g(r(\xi-\epsilon))<g(\frac\Lambda\mu).
		\label{epsilon}
	\end{equation}
	Hence, for every $\omega\in \Omega_1$, there is a $T=T(\omega)>0$ such that
	\begin{equation}
		\label{thm5.1-proof-contradicts-1}
		Q(t,\omega)\le\xi-\epsilon,\qquad\text{whenever }t\ge T(\omega).
	\end{equation}
 However, for each pair of $(Q,I)\in\varGamma_0'\cap\{Q\le\xi-\epsilon\}$, by 4) of Lemma 2.6 and the property of quadratic functions,
 \begin{equation}
	\tilde f(Q,I)=g(h(Q,I))\ge g(\frac\Lambda\mu)\wedge g(r(Q)).
	\label{thm5.1-f1}
 \end{equation}
 where by again the property of quadratic functions,
 \begin{equation}
	g(r(Q))\ge g(r(\xi-\epsilon))\wedge g(r(0))=g(r(\xi-\epsilon))\wedge g(\frac\Lambda\mu).
	\label{thm5.1-f2}
 \end{equation}
	Therefore by (\ref{epsilon}), (\ref{thm5.1-proof-contradicts-1}), (\ref{thm5.1-f1}) and (\ref{thm5.1-f2}), we have
	\begin{equation}
		\label{thm5.1-proof-f}
		\tilde f(Q(t,\omega), I(t,\omega))\ge g(r(\xi-\epsilon)),\qquad \text{whenever  } t\ge T(\omega).
	\end{equation}
	Moreover, by the large number theorem for martingales, there is an $\Omega_2$ with $\mathbb P(\Omega_2)=1$ such that for every $\omega\in\Omega_2$,
	\begin{equation}
		\lim_{t\to\infty}\frac1 t\int^t_0\sigma e^{-\alpha I(s,\omega)}(\frac\Lambda\mu-Q(s,\omega))dB(s)(\omega)=0.
		\label{thm5.1-proof-large}
	\end{equation}
	Now, fix any $\omega\in \Omega_1\cap\Omega_2$. It then follows form (\ref{ito-extinction1}) and (\ref{thm5.1-proof-f}), for $t\ge T(\omega)$,
	\begin{align}
		\ln(I(t,\omega))\ge&\ln(I(0,\omega))+\int^{T(\omega)}_0\tilde f(Q(s,\omega),I(s,\omega))ds+(t-T(\omega))g(r(\xi-\epsilon))\notag\\ &+\int^t_0\sigma e^{-\alpha I(s,\omega)}(\frac\Lambda\mu-Q(s,\omega))dB(s)(\omega).
		\label{thm3.7-ito}
	\end{align}
	Combining (\ref{thm3.7-ito}), (\ref{thm5.1-proof-large}) and 3) in Lemma 2.6 leads to
	\[\liminf_{t\to\infty}\frac 1 t \ln(I(t,\omega))\ge g(r(\xi-\epsilon))>0,\]
	thus
	\[\lim_{t\to\infty}I(t,\omega)=\infty\qquad\omega\in\Omega_1\cap\Omega_2.\]
	This contradicts to (\ref{thm5.1-proof-contradicts-1}). Therefore we must have the desired weak persistence of $Q(t)$.

	Now we can prove the weak persistence of the disease $I(t)$. Set $\xi'=(1+\frac\gamma\mu)^{-1}\xi$. If the assertion (\ref{thm5.3-conclusion1}) is false, then there is a sufficiently small $\epsilon>0$ such that
	$\mathbb P(\Omega_4)>\epsilon$, where $\Omega_4=\{\omega|\limsup_{t\to\infty}I(t,\omega)\le\xi'-2\epsilon\}.$
	Hence, by (\ref{main}), for every $\omega\in \Omega_4$, for large enough $t$, we have
	\[dR(t,\omega)\le(\gamma(\xi'-\epsilon)-\mu R(t,\omega))dt, \]
	which yields
	\[\limsup_{t\to\infty}R(t,\omega)\le\frac{\gamma}{\mu}(\xi'-\epsilon).\]
	Combining this with the definition of $\Omega_4$, we have
	\begin{equation}
		\limsup_{t\to\infty}(I(t,\omega)+R(t,\omega))\le(1+\frac\gamma\mu)(\xi'-\epsilon)<\xi,\omega\in \Omega_4.
		\label{thm5.1-contradiction-1}
	\end{equation}
	However, by (\ref{thm5.1-conclusion-1}) and the fact that
	\[\lim_{t\to\infty}(Q(t)-I(t)-R(t))=0,\]
	we must have
	\begin{equation}
		\limsup_{t\to\infty}(I(t)+R(t))\ge\xi\qquad \text{a.s.}
		\label{thm5.1-contradiction-2}
	\end{equation}
	This contradiction (\ref{thm5.1-contradiction-1}). Therefore we must have the desired weak persistence of $I(t)$. The proof is done.
\end{proof}
Now we discuss the mean persistence of $I(t)$. For convenience, we introduce the following notation. For a continuous stochastic process $y(t)$, let
\[\langle y(t)\rangle=\frac 1 t \int_0^t y(s)ds,\]
We have the following result for $\langle I(t)\rangle$.
\begin{theorem}
	If
	\begin{equation}
		R_0^S>1+\frac{\alpha\beta\Lambda^2}{4\mu^2(\mu+\gamma)},
		\label{thm3.10-condition}
	\end{equation}
	then for any given initial value $(S(0),I(0),R(0))\in\varGamma$, the solution of the SDE(\ref{main}) obeys
	\begin{equation}
		\liminf_{t\to\infty}\langle I(t) \rangle >\frac{\mu}{\beta}\left(R_0^S-1-\frac{\alpha\beta\Lambda^2}{4\mu^2(\mu+\gamma)}\right)>0.
		\label{thm3.10-conclusion}
	\end{equation}
	\label{thm-mean-persistence}
\end{theorem}
\begin{proof}
	By (\ref{ito-extinction1}), we have
	\begin{align}
		\frac{\ln I(t)-\ln I(0)}t-\frac 1 t \int_0^t\sigma e^{-\alpha I}SdB(s)={}&\beta\langle e^{-\alpha I}S\rangle-(\mu+\gamma)-\frac 1 2\sigma^2\langle e^{-2\alpha I}S^2\rangle \notag\\
		\ge{}&\beta\langle (1-\alpha I)S\rangle-(\mu+\gamma)-\frac 1 2 \sigma^2\frac{\Lambda^2}{\mu^2}\notag\\
		={}&\beta\langle S\rangle-\alpha\beta\langle IS\rangle-(\mu+\gamma)-\frac 1 2 \sigma^2\frac{\Lambda^2}{\mu^2}\notag\\
		\ge{}&\beta\langle S\rangle-\frac{\alpha\beta\Lambda^2}{4\mu^2}-(\mu+\gamma)-\frac 1 2 \sigma^2\frac{\Lambda^2}{\mu^2},\label{thm3.10 <I>}
	\end{align}
	where the first inequality in (\ref{thm3.10 <I>}) is due to $e^{-x}\ge 1-x,\forall x\in \mathbb R$ and the last inequality is derived by $I(t)+S(t)\le \frac\Lambda\mu$. From (\ref{main}) we have
	\begin{equation}
		\frac {S(t)-S(0)}t+\frac{I(t)-I(0)}t=\Lambda-\mu \langle S\rangle -(\mu+\gamma)\langle I\rangle.
		\label{thm3.10 <S>}
	\end{equation}
	Substituting (\ref{thm3.10 <S>}) into (\ref{thm3.10 <I>}), we get
	\begin{equation}
		\frac{\beta(\mu+\gamma)}{\mu}\langle I(t)\rangle\ge \frac {\beta\Lambda}{\mu}-\frac{\alpha\beta\Lambda^2}{4\mu^2}-(\mu+\gamma)-\frac{\sigma^2\Lambda^2}{2\mu^2}-\frac {\ln I(t)}t+\phi(t),
		\label{thm3.10-inequality}
	\end{equation}
	where
	\begin{equation}
		\phi(t)=\frac 1 t\int_0^t \sigma e^{-\alpha I}SdB(s)-\frac{\beta(S(t)-S(0)+I(t)-I(0))-\mu\ln I(0)}{\mu t}.
	\end{equation}
	Clearly $\limsup_{t\to\infty}\phi(t)=0$. Also since $I(t)\le\frac\Lambda\mu$, we have $\limsup_{t\to\infty} \frac{\ln I(t)}t\le 0$. Then (\ref{thm3.10-inequality}) becomes
	\[\frac{\beta}{\mu}\liminf_{t\to\infty}\langle I(t)\rangle\ge R_0^S-\frac{\alpha\beta\Lambda^2}{4\mu^2(\mu+\gamma)}-1>0,\]
	and the proof is complete.
\end{proof}
We give the weak persistence and mean persistence of $I(t)$ separately in Theorem 2.5 and Theorem 2.7. Note that the condition for the mean persistence of the disease is much stronger than the condition for the weak persistence of it, since the former one is $R_0^S>1$ while the latter one is
\[R_0^S>1+\frac{\alpha\beta\Lambda^2}{4\mu^2(\mu+\gamma)}.\]

However, although these two conditions are different, they are both related to $R_0^S$, which means $R_0^S$ can be used as a major criterion for the persistence of the disease.

\section{Numerical simulations}

In this section, we use the stochastic Runge-Kutta method in \cite{Numerical} to simulate the stochastic model (\ref{main}) and the corresponding deterministic model. We initially verify our theoretical results, and then consider the effect of stochastic perturbation on infections  in  which the theoretical results do not cover. To illustrate the impact of the stochastic perturbation, we perform simulations for the infectious individuals for both the stochastic model (\ref{main}) and the corresponding deterministic model (\ref{deterministic}) by freely choosing parameter values.

We initially choose parameter values as  in Fig.\ref{example1} where  $R_0^S=0.88<1$ and $\sigma^2< {\mu\beta}/{\Lambda}$, we then show the infected individuals go to extinction with probability one (shown in  Fig.\ref{example1}), according to Theorem 2.3.
It is worth noting that for the deterministic model  the disease  persists for  $R_0^D=1.36>1$,  we can show that small noise does not change the disease persistence according to Theorem 2.5 (shown  in Fig.\ref{example3}), while the $I(t)$ goes to zero  in the stochastic model with large noise (shown in  Fig.\ref{example2}). This  illustrates great  environmental noise can cause disease to go to extinction. In particular,  solutions  in the stochastic model fluctuate around the corresponding counterpart in the deterministic model.

Note that when $R_0^S<1$, Theorem 2.3 and 2.4 show that the disease will die out if $\sigma^2<\frac{\mu\beta}{\Lambda}$ or $\sigma^2>\frac{\beta^2}{2(\mu+\gamma)}$.  Then there are two situations:
 $R_0^S<1$ with $\frac{\mu\beta}{\Lambda}\le\sigma^2\le\frac{\beta^2}{2(\mu+\gamma)}$ and $R_0^S=1$,  under which
we do not know what the solutions approach.   Hence we will pay attention on the two  situations and give our results through numerical simulations.

\begin{figure}[!htbp]
    \centering
    \subfloat[]{\label{example1}\includegraphics[width=0.32\linewidth,trim= 10 0 80 30,clip]{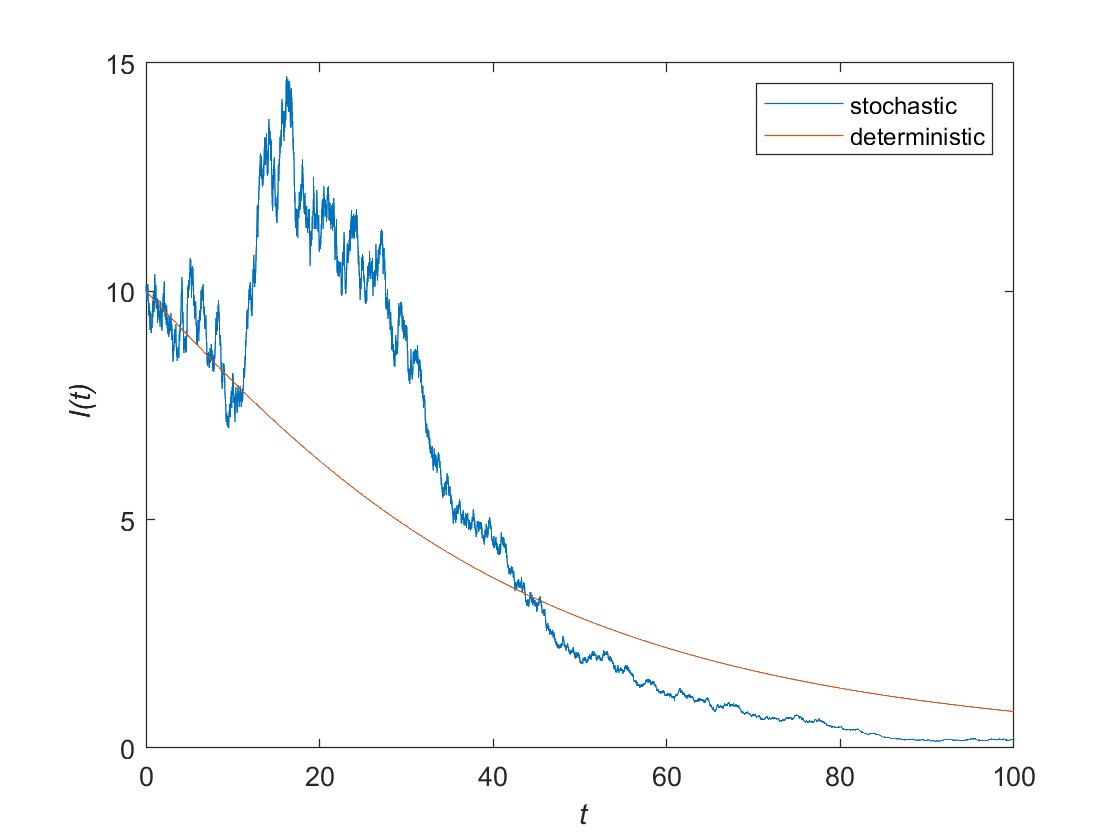}}
    \subfloat[]{\label{example2}\includegraphics[width=0.32\linewidth,trim= 10 0 80 30,clip]{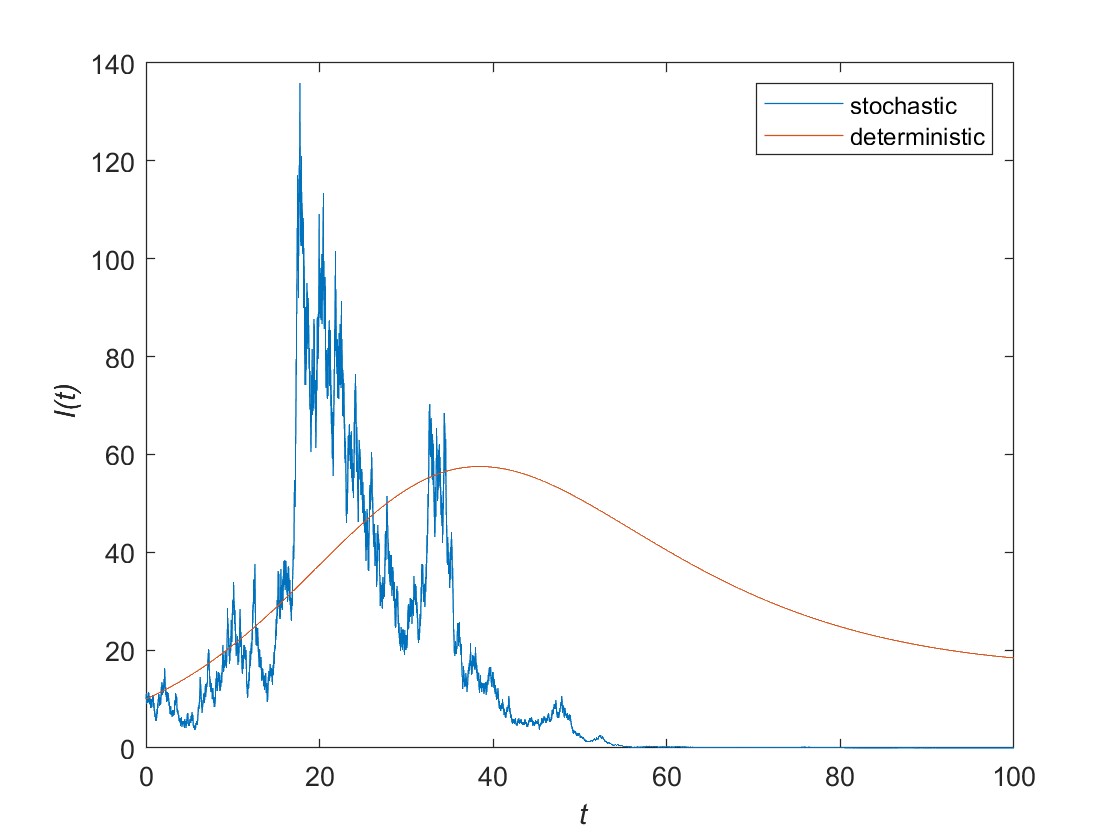}}
	\subfloat[]{\label{example3}\includegraphics[width=0.32\linewidth,trim= 10 0 80 30,clip]{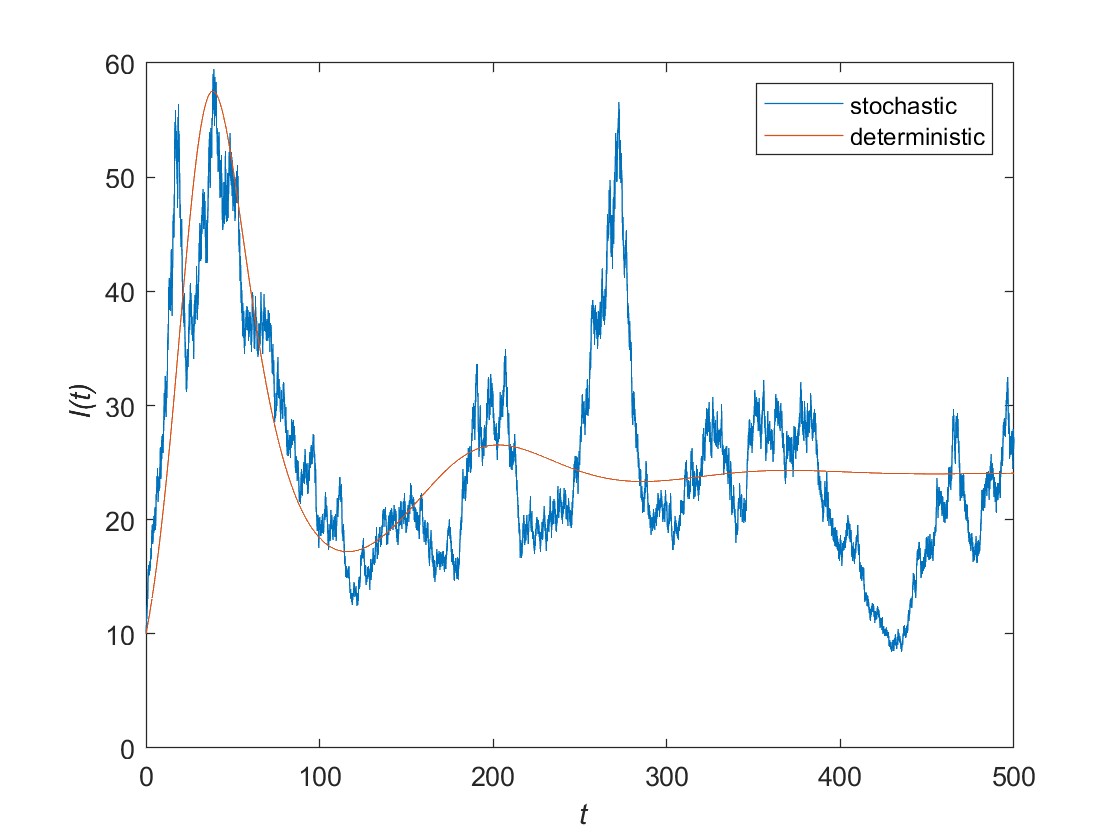}}
	\caption{Computer simulation of the path $I(t)$ for the stochastic model (\ref{main}) and the corresponding deterministic model (\ref{deterministic}) with $S(0)=1000, I(0)=10, R(0)=0, \alpha=1\times 10^{-4}, \Lambda=20, \mu=0.02, \gamma=0.2$ and (a) $\sigma=1\times 10^{-4},\beta=2\times 10^{-4}$, (b)$\sigma=4.5\times 10^{-4}, 3\times 10^{-4}$, (c)$\sigma=1\times 10^{-4}, \beta=3\times 10^{-4}$}
	\label{example1-2}
\end{figure}

%

\begin{figure}[!htbp]
    \centering
    \subfloat[]{\label{example4}\includegraphics[width=0.45\linewidth,trim= 10 0 70 30,clip]{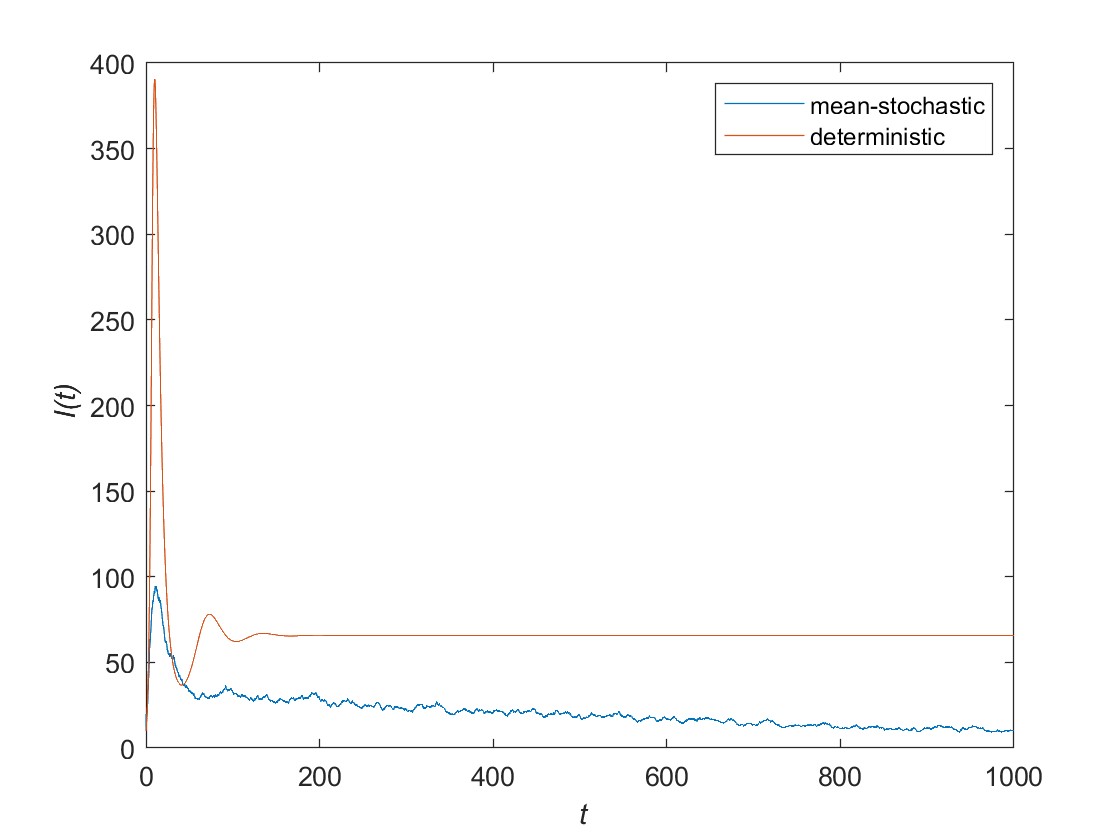}}\hspace{30pt}
    \subfloat[]{\label{example5}\includegraphics[width=0.45\linewidth,trim= 10 0 70 30,clip]{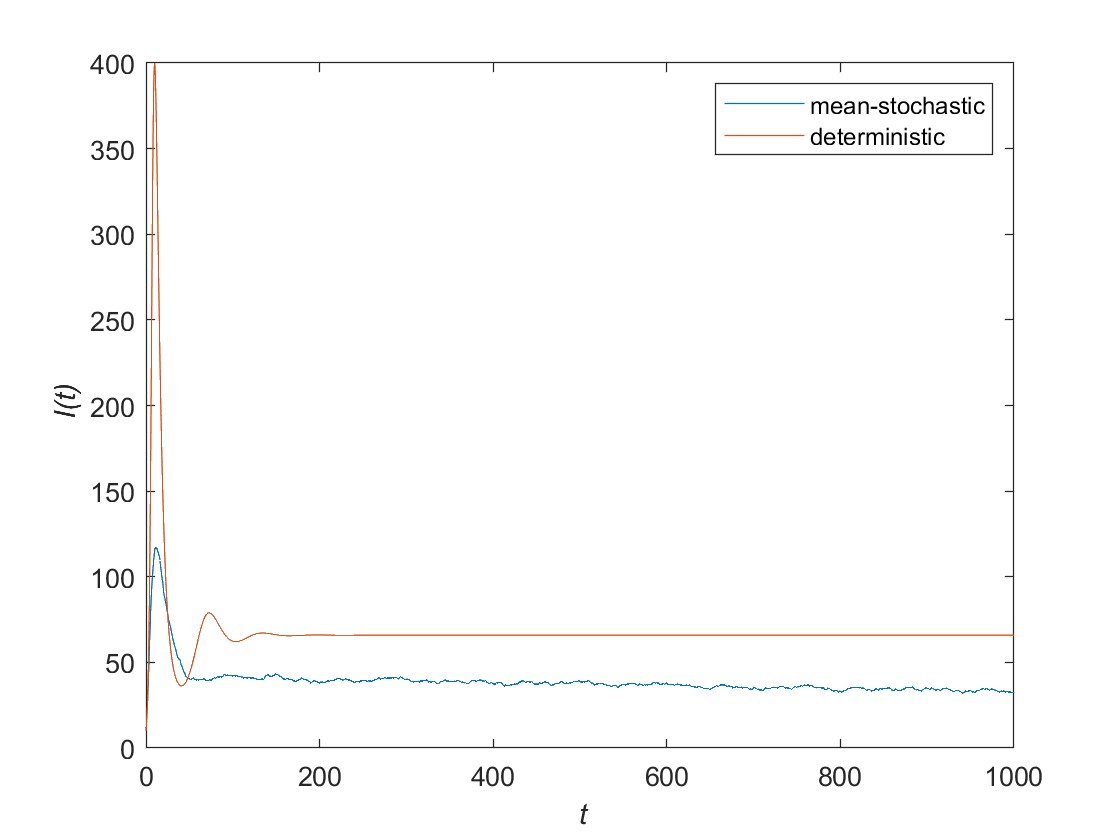}}
	\caption{The mean value of $I(t)$ in $n$ simulations for the stochastic model (\ref{main}) and $I(t)$ in the corresponding deterministic model (\ref{deterministic}) when $S(0)=1000, I(0)=10, R(0)=0, \alpha=1\times 10^{-4}, \Lambda=20, \mu=0.02, \gamma=0.2, \beta=8\times 10^{-4}$ and (a)$n=1000, \sigma=11.3\times 10^{-4}$ (b) $ n=5000,  \sigma=10.771\times 10^{-4}$}
\end{figure}


To further examine the asymptotical property for the situation $R_0^S<1$ and $\frac{\mu\beta}{\Lambda}\le\sigma^2\le\frac{\beta^2}{2(\mu+\gamma)}$, we carry out 1000 simulations  for the stochastic model (\ref{main}).
We choose the parameters as in Fig.\ref{example4} in which $R_0^S=0.90$ and $\frac{\mu\beta}{\Lambda}\le\sigma^2\le\frac{\beta^2}{2(\mu+\gamma)}$. Letting $\epsilon=0.0001$, we found that in each simulation, there exists a $T$ such that
$I(t)<\epsilon, \forall t\geq T.$  It follows from
Fig.\ref{example4} that the mean value of $I(t)$ tends to zero. So we can draw the conclusion that the disease will die out almost surely if $R_0^S<1$, and  $\frac{\mu\beta}{\Lambda}\le\sigma^2\le\frac{\beta^2}{2(\mu+\gamma)}$.
For the scenario of $R_0^S=1$, we take the values as in Fig.\ref{example5}. Again, we run the 5000 simulations and show the mean value of $I(t)$ in Fig.\ref{example5}, we observe there are 3963 simulations in which the value of $I(t)$ fell below $\epsilon$ in a give time. Hence, for $R_0^S=1$, mean value of $I(t)$ approaches to zero with high probability (i.e., disease has the great tendency to die out).


To examine the effect of media impact on disease infection we plot the variation in the mean value of $I(t)$ with parameter $\alpha$ (shown in Fig.\ref{example6.1} and Fig.\ref{example6.2}). One can see that increasing the parameter value of  $\alpha$  leads to low infection. This means mass media induced behaviour changes (reduced incidence) play a vital role in control of disease infection in the stochastic environment, which agrees well with those concluded from the deterministic models  \cite{media-model2.3,media-model2.1}.


\begin{figure}[!htbp]
    \centering
    \subfloat[]{\label{example6.1}\includegraphics[width=0.45\linewidth,trim= 10 0 70 30,clip]{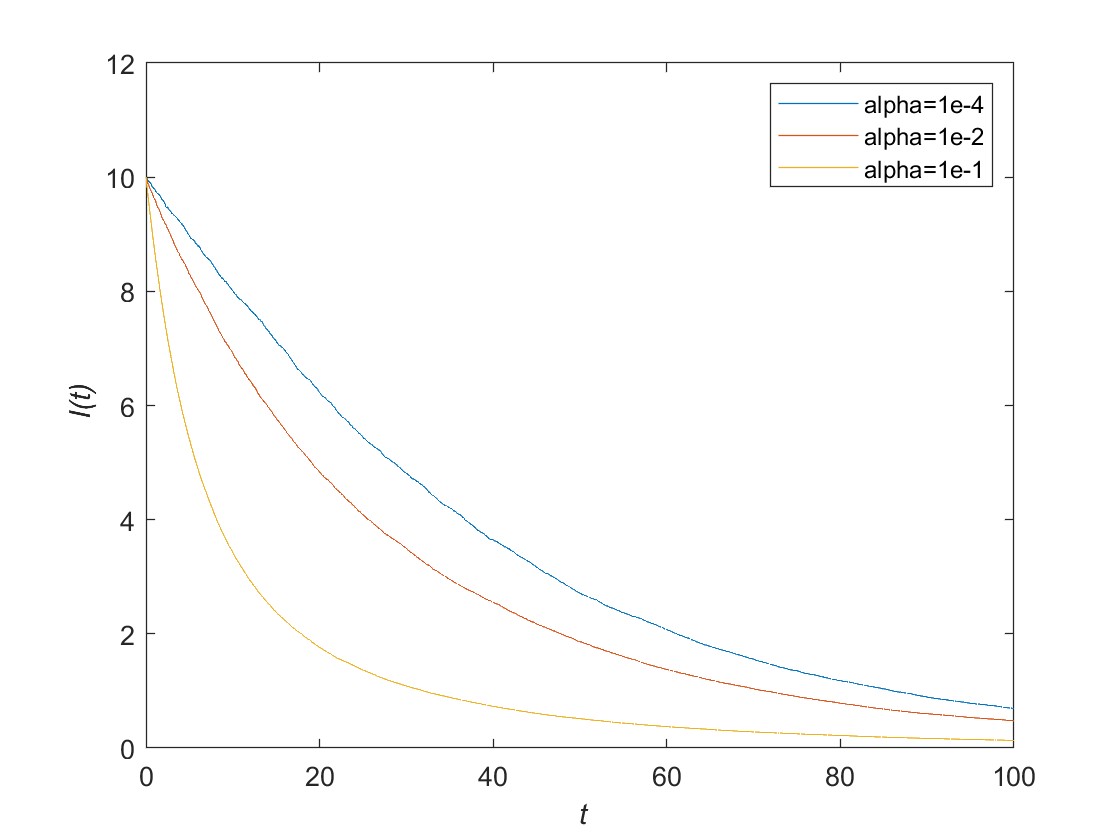}}\hspace{8pt}
    \subfloat[]{\label{example6.2}\includegraphics[width=0.45\linewidth,trim= 10 0 70 30,clip]{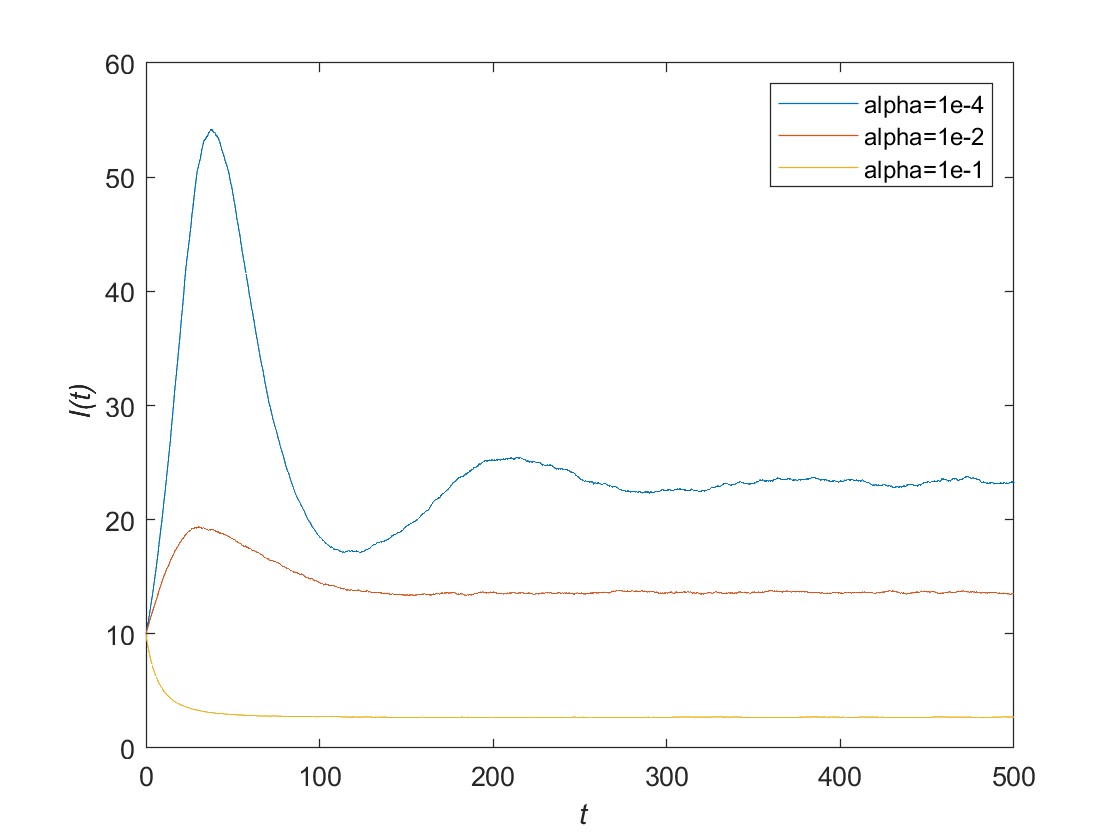}}
    \caption{The mean value of $I(t)$ in 1000 simulations for the stochastic model (\ref{main}) under different $\alpha$ values with $S(0)=1000, I(0)=10, R(0)=0, \Lambda=20, \mu=0.02, \gamma=0.2, \sigma=1\times 10^{-4}$ and (a) $\beta=2\times 10^{-4}$, $R_0^S=0.88$, (b)$\beta=3\times 10^{-4}, R_0^S=1.34$.}
	\label{example6}
\end{figure}

\section{Discussion}
In the real world, biological and epidemiological phenomenon are always affected by the environmental noises. As a result, stochastic models may provide more natural description and consequently produce more valuable results, compared to the deterministic counterparts
\cite{SDE1,SDE2,SDE3,SDE4,SDE5,SDE6,SDE7,SDE8,SDE9}.  In this study considering environmental noises, we investigated the transmission dynamics of an epidemic model with media effects. We initially investigated the existence and uniqueness of solutions of the stochastic system (\ref{main}), then examined conditions under which the disease dies out or persists. In particular, we obtained  the disease goes to extinct if $R_0^S<1$, $\sigma^2<\frac{\mu\beta}{\Lambda}$ or $\sigma^2>\frac{\beta^2}{2(\mu+\gamma)}$ (then $R_0^S<1$ by the remark below Theorem 2.4), while the system is weak persistent for  $R_0^S>1$. This indicates that the basic reproduction number can act as the threshold value given small noise, which is similar to the threshold level of $R_0^D$ in the deterministic system. By comparing the  $R_0^D$ with $R_0^S$ we know that $R_0^S\le R_0^D$, which means with environmental noises disease goes to extinction more likely. Further,
large noise can induce the disease extinction with probability of 1,  implies that environmental noises can not be ignored when investigating threshold dynamics and estimating the threshold level (the basic reproduction number).

In terms of media impacts under environmental noises, we obtained the basic reproduction number $R_0^S<1$ is also independent of the media-related  parameter $\alpha$, that is, inclusion of media induced behaviour changes does not affect the threshold itself, which is similar to the conclusion of the deterministic models.  However, numerical simulations suggest that media impacts induce the disease infection decline, which is also verified by the theorem 2.7  that due to media impact strong persistence of the stochastic system becomes less likely.

Finally, we would like to mention the limitations of our work. We only gave the numerical results under the conditions
\[R_0^S<1,\frac{\mu\beta}{\Lambda}\le\sigma^2\le\frac{\beta^2}{2(\mu+\gamma)}\quad\text{and}\quad R_0^S=1.\]
It is rather interesting that whether we can analytically prove that the disease will go extinct under these conditions. In fact, a similar question proposed in \cite{SDE1} is still an open question. Futhermore, in this study we only introduce the white noises to the system (\ref{deterministic}). It is interesting to investigate the effects of impulsive perturbations on system (\ref{deterministic}). These problems will be the subjects of our future work.

{\bf Acknowledgments}  This work is supported by the National Natural Science Foundation of China (NSFC, 12220101001, 12031010).


\begin{thebibliography}{99}
	\bibitem{model}Kermack W O, McKendrick A G. A contribution to the mathematical theory of epidemics[J]. Proceedings of the royal society of london. Series A, Containing papers of a mathematical and physical character, 1927, 115(772): 700-721.
	\bibitem{Media-endemic1}Cui J, Sun Y, Zhu H. The impact of media on the control of infectious diseases[J]. Journal of dynamics and differential equations, 2008, 20(1): 31-53.
	\bibitem{Media-endemic2}Cui J A, Tao X, Zhu H. An SIS infection model incorporating media coverage[J]. The Rocky Mountain Journal of Mathematics, 2008: 1323-1334.
	\bibitem{Media-endemic3}Li Y, Cui J. The effect of constant and pulse vaccination on SIS epidemic models incorporating media coverage[J]. Communications in Nonlinear Science and Numerical Simulation, 2009, 14(5): 2353-2365.
	\bibitem{Media-endemic4}Sun C, Yang W, Arino J, et al. Effect of media-induced social distancing on disease transmission in a two patch setting[J]. Mathematical biosciences, 2011, 230(2): 87-95.
	\bibitem{Media-Covid1}Depoux A, Martin S, Karafillakis E, et al. The pandemic of social media panic travels faster than the COVID-19 outbreak[J]. Journal of travel medicine, 2020, 27(3): taaa031.
	\bibitem{Media-Covid2}Gonzmlez-Padilla D A, Tortolero-Blanco L. Social media influence in the COVID-19 Pandemic[J]. International braz j urol, 2020, 46: 120-124.
	\bibitem{media-model1}Liu R, Wu J, Zhu H. Media/psychological impact on multiple outbreaks of emerging infectious diseases[J]. Computational and Mathematical Methods in Medicine, 2007, 8(3): 153-164.
	\bibitem{media-model2.1}Cui J, Sun Y, Zhu H. The impact of media on the control of infectious diseases[J]. Journal of dynamics and differential equations, 2008, 20(1): 31-53.
	\bibitem{media-model2.2}Wang A, Xiao Y. A Filippov system describing media effects on the spread of infectious diseases[J]. Nonlinear Analysis: Hybrid Systems, 2014, 11: 84-97.
	\bibitem{media-model2.3}Song P, Xiao Y. Global hopf bifurcation of a delayed equation describing the lag effect of media impact on the spread of infectious disease[J]. Journal of mathematical biology, 2018, 76(5): 1249-1267.
	\bibitem{epidemic 1}Zaman G, Kang Y H, Jung I H. Stability analysis and optimal vaccination of an SIR epidemic model[J]. BioSystems, 2008, 93(3): 240-249.
	\bibitem{epidemic 2}Thornley S, Bullen C, Roberts M. Hepatitis B in a high prevalence New Zealand population: a mathematical model applied to infection control policy[J]. Journal of Theoretical Biology, 2008, 254(3): 599-603.
	\bibitem{epidemic 3}Zou L, Zhang W, Ruan S. Modeling the transmission dynamics and control of hepatitis B virus in China[J]. Journal of theoretical biology, 2010, 262(2): 330-338.
	\bibitem{SDE1}Tornatore E, Buccellato S M, Vetro P. Stability of a stochastic SIR system[J]. Physica A: Statistical Mechanics and its Applications, 2005, 354: 111-126.
	\bibitem{SDE2}Gray A, Greenhalgh D, Hu L, et al. A stochastic differential equation SIS epidemic model[J]. SIAM Journal on Applied Mathematics, 2011, 71(3): 876-902.
	\bibitem{SDE3}Zhao Y, Jiang D. The threshold of a stochastic SIRS epidemic model with saturated incidence[J]. Applied Mathematics Letters, 2014, 34: 90-93.
	\bibitem{SDE4}Zhang X B, Wang X D, Huo H F. Extinction and stationary distribution of a stochastic SIRS epidemic model with standard incidence rate and partial immunity[J]. Physica A: Statistical Mechanics and its Applications, 2019, 531: 121548.
	\bibitem{SDE5}Hussain G, Khan A, Zahri M, et al. Ergodic stationary distribution of stochastic epidemic model for HBV with double saturated incidence rates and vaccination[J]. Chaos, Solitons \& Fractals, 2022, 160: 112195.
	\bibitem{SDE6}Berrhazi B, El Fatini M, Lahrouz A, et al. A stochastic SIRS epidemic model with a general awareness-induced incidence[J]. Physica A: Statistical Mechanics and its Applications, 2018, 512: 968-980.
	\bibitem{SDE7}Cao Z, Feng W, Wen X, et al. Dynamics of a stochastic SIQR epidemic model with standard incidence[J]. Physica A: Statistical Mechanics and its Applications, 2019, 527: 121180.
	\bibitem{SDE8}Khan T, Khan A, Zaman G. The extinction and persistence of the stochastic hepatitis B epidemic model[J]. Chaos, Solitons \& Fractals, 2018, 108: 123-128.
	\bibitem{SDE9}Jiang D, Ji C, Shi N, et al. The long time behavior of DI SIR epidemic model with stochastic perturbation[J]. Journal of Mathematical Analysis and Applications, 2010, 372(1): 162-180.
	\bibitem{SDE}X.Mao, Stochastic Differential Equations and Applications, 2nd ed., Horwood, Chichester, UK, 2008
	\bibitem{Numerical}Milstein G N, Tret'yakov M V. Mean-square numerical methods for stochastic differential equations with small noises[J]. SIAM Journal on Scientific Computing, 1997, 18(4): 1067-1087.
\end{thebibliography}
\end{document}